\documentclass[10pt,letterpaper]{article}
\usepackage[margin=1in]{geometry}
\usepackage{ifpdf}
\ifpdf
    \usepackage[pdftex]{graphicx}
    \graphicspath{{figs/}{matlab/}}   
    \usepackage[update]{epstopdf}
\else
	\usepackage{graphicx}
	\graphicspath{{figs/}{matlab/}}   
\fi
\usepackage{pdflscape}
\usepackage{marvosym}
\usepackage{enumitem,color}
\usepackage{amsthm}
\newtheorem{theorem}{Theorem}

\newtheorem{lemma}{Lemma}
\newtheorem{corollary}{Corollary}

\usepackage{hyperref}
\usepackage{array}
\usepackage{amsmath,amsthm}
\usepackage{amsfonts}
\usepackage{amssymb}
\usepackage{cite,setspace}
\usepackage{breqn}
\usepackage{multirow}

\newtheorem{definition}{\textbf{Definition}}

\newcommand{\Expt}{\mbox{${\mathbb E}$} }

\begin{document}

\title{On the Storage Cost of Private Information Retrieval}
\author{Chao Tian}
\maketitle

\begin{abstract}
We consider the fundamental tradeoff between the storage cost and the download cost in private information retrieval systems, without any explicit structural restrictions on the storage codes, such as maximum distance separable codes or uncoded storage. Two novel outer bounds are provided, which have the following implications. When the messages are stored without any redundancy across the databases, the optimal PIR strategy is to download all the messages; on the other hand, for PIR capacity-achieving codes, each database can reduce the storage cost, from storing all the messages, by no more than one message on average. We then focus on the two-message two-database case, and show that a stronger outer bound can be derived through a novel pseudo-message technique. This stronger outer bound suggests that a precise characterization of the storage-download tradeoff may require non-Shannon type inequalities, or at least more sophisticated bounding techniques. 
\end{abstract}

\section{Introduction}

Private information retrieval \cite{chor1995private}, or simply referred to as PIR, is a fundamental privacy-preserving information processing primitive. PIR has deep connections to other well-known communication and coding problems such as locally decodable codes \cite{goldreich2002lower,sun2018capacityldc} and interference alignment \cite{sun2016blind}. The PIR capacity, i.e., the inverse of the minimum download cost, was recently characterized by Sun and Jafar \cite{sun2017PIRcapacity}, and the PIR capacities under other variations have also been considered \cite{Tajeddine_Rouayheb,Sun_Jafar_TPIR,sun2018capacitysymmetric,Banawan_Ulukus,banawan2018capacityByzantine,Sun_Jafar_MPIR,Sun_Jafar_MDSTPIR,FREIJ_HOLLANTI,kumar2019achieving,
banawan2018capacity,Banawan_Ulukus_Multimessage,chen2017capacity,Wei_Banawan_Ulukus_Side,tian2018capacity,zhou2019capacity}. 

Since the databases from which the information is retrieved are basically storage nodes, designing efficient storage strategies in PIR systems is of significant importance, and the problem has received considerable attention recently. Some of the existing efforts assume certain specific coding structures in the storage side, such as using maximum distance separable (MDS) codes \cite{banawan2018capacity,Tajeddine_Rouayheb} or in an uncoded form \cite{attia2018capacity}. In two recent works \cite{Tian_Sun_Chen_Storage,Sun_Jafar_MPIR}, the tradeoff was considered without any structural constraints on the storage or retrieval codes  for the special case of two databases and two messages, and it was found that non-linear codes can provide further improvement over linear codes. Other notable efforts can be found in \cite{Rao_Vardy,Chan_Ho_Yamamoto,sun2019breaking,fazeli2015codes,BanawanITW19} and references therein. 

Despite these efforts, our understanding of the fundamental tradeoff between the storage cost $\alpha$ and the download cost $\beta$ is still quite limited, mainly due to the lack of general information theoretic converse results when the storage codes are not required to follow any potentially restricting structural constraint. In this work, we initialize such an effort and derive several information theoretical outer bounds of the fundamental tradeoff between the storage cost and the download cost. Instead of attempting to characterize the complete optimal tradeoff curve, in this work we focus on the two extreme points: the point when the storage cost is minimal, and the point when the download cost is minimal. For the former, the question is when the storage has no redundancy, what the minimum download cost can be; for the latter, the question is for PIR capacity-achieving codes, what the minimum storage cost can be; see Fig. \ref{fig:concept} for an illustration. In other words, in this work, our goal is to identity the two anchor points of the fundamental tradeoff curve, from which the general tradeoff can hopefully be further built through subsequent efforts.

Our main result is a precise characterization of the first extreme point (no redundancy in storage), and an approximate characterization of the second (capacity-achieving PIR codes). More precisely, by providing two novel outer bounds, we show that for the former, the optimal PIR strategy is in fact to download all the messages, whereas for the latter, i.e., for PIR capacity-achieving codes, each database can reduce the storage cost by no more than one message on average, compared to the simple strategy of storing every message. In order to better understand the second extreme point, we then focus on the two-message two-database case, and show that a stronger outer bound can be derived. This stronger bound requires a novel pseudo-message technique, the origin of which can be traced back to Ozarow's bounding technique \cite{Ozarow:80} of extending the original probability space through Markov coupling. This technique has been developed significantly over the years and found many applications, e.g., \cite{wagner2008improved,fu2002rate,tian2009approximating}, and it is also the technique based on which all known non-Shannon type inequalities were derived in the literature \cite{zhang1997non,dougherty2011non,gurpinar2019use}. In the specific context of PIR, however, the value of this refined outer bound is the following: it strongly suggests that a precise characterization of the storage-download tradeoff may require non-Shannon type inequalities, or at least certain more sophisticated bounding techniques.

\begin{figure}[t]
\centering
\includegraphics[width=0.35\textwidth]{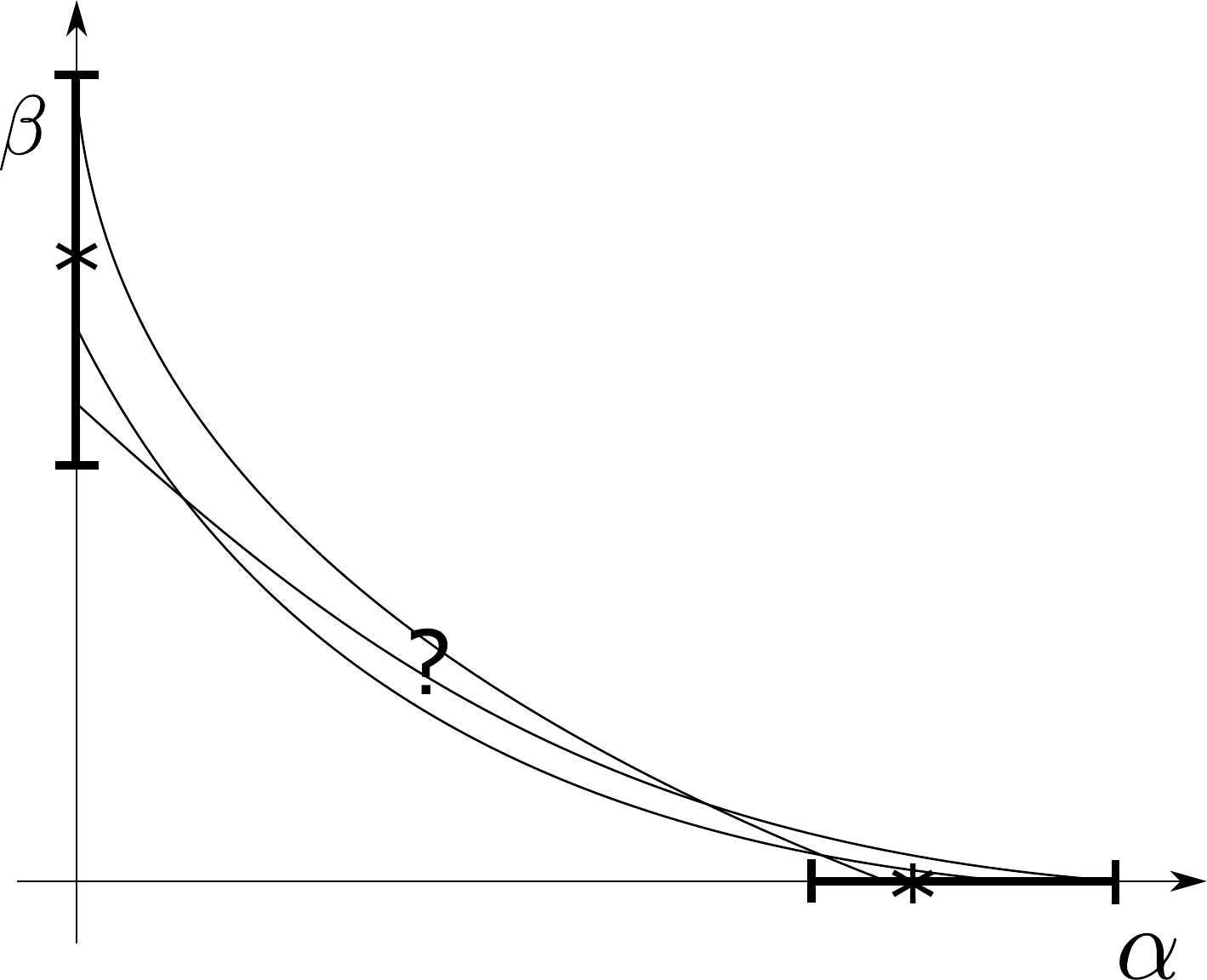}
\caption{Instead of characterizing the optimal tradeoff curve between the storage cost $\alpha$ and download cost $\beta$, we focus on finding (approximate) characterizations of the minimum $\alpha$ on the horizontal axis and the minimum $\beta$ on the vertical axis.\label{fig:concept}}
\end{figure}

\section{Problem Definition}

The problem setting of private information retrieval is well known, see e.g., \cite{sun2017PIRcapacity}, however in order to fix the notation and derive the outer bounds in question, we need to first introduce a rigorous version of the problem definition. We then provide the formal definitions of the operational download cost and storage cost, together with the informational version, in order to unify the different conventions in the literature and avoid confusion in subsequent discussions. 

\subsection{Encoding and Decoding Functions}
There are a total of $K$ messages $W_1,W_2,\ldots,W_K$ in the system, and there are a total of $N$ databases that the messages are to be retrieved; the messages $(W_i,W_{i+1},\ldots,W_j)$ will sometimes be written together as $W_{i:j}$ for conciseness. Denote the set of possible queries that server-$n$ can accommodate as $\mathcal{Q}_n$, and denote its cardinality as $|\mathcal{Q}_n|$. The cardinality of a set $\mathcal{A}$ will be similarly denoted as $|\mathcal{A}|$ in the rest of the paper. Assume that a random key $\mathsf{F}$ is uniformly distributed on a certain finite set $\mathcal{F}$, which is used by the user to produce the (random) queries to the $N$ databases. The servers, after receiving the queries for message-$k$, denoted as $Q_n^{[k]}$, will reply with an answer $A_n^{[k]}$. A message $W_k$ consists of $L$ symbols, each symbol belonging to a finite alphabet $\mathcal{X}$. The messages are mutually independent, each of which is uniformly distributed on $\mathcal{X}^L$.  We further allow the query answers to be represented as a variable-length vector, whose elements are in the finite alphabet $\mathcal{Y}$. A mathematically precise description of the problem is given next via a set of encoding and decoding functions. 

\begin{definition}
\label{def:problem}
An $N$-server private information retrieval (PIR) storage code for $K$ messages, each of $L$-symbols in the alphabet $\mathcal{X}$, consists of
\begin{enumerate}[itemsep=0pt]
\item $N$ storage encoding functions:
\begin{align}
\Phi_{n}: \mathcal{X}^{KL}\rightarrow \mathcal{S}_n,\qquad n\in \{1,2,...,N\},\label{eqn:storage}
\end{align}
i.e., the stored information at database-$n$ is $S_n=\Phi_{n}(W_{1:K})$;
\item $N$ query functions: 
\begin{align}
&\phi_n: \{1,2,\ldots,K\}\times \mathcal{F}\rightarrow \mathcal{Q}_n,\qquad n\in \{1,2,...,N\},
\end{align}
i.e., the user chooses the query $Q^{[k]}_n=\phi_n(k,\mathsf{F})$ for server-$n$, using the index of the desired message and the random key $\mathsf{F}$;
\item $N$ answer length functions: 
\begin{align}\ell_n: \mathcal{Q}_n\rightarrow \{0,1,2,\ldots\},\qquad n\in \{1,2,...,N\}, 
\end{align} 
i.e., the length of the answer at each server, a non-negative integer, is a deterministic function of the query, but not the particular realization of the messages; 

\item $N$ answer functions: 
\begin{align}
\varphi_{n}: \mathcal{Q}_n\times \mathcal{S}_n\rightarrow \mathcal{Y}^{\ell_n},\qquad n\in\{1,2,...,N\},\label{eqn:answers}
\end{align}
where $\ell_n=\ell_n(q_n)$ with $q_n\in \mathcal{Q}_n$ being the (random) query for server-$n$, $\mathcal{Y}$ is the coded symbol alphabet, and in the sequel we write the query answer as $A^{[k]}_n\triangleq\varphi_{n}(Q^{[k]}_n,S_n)$ when the message index $k$ is relevant;
\item A reconstruction function using the answers from the servers together with the desired message index and the random key: 
\begin{align}
\psi: \prod_{n=1}^{N} \mathcal{Y}^{\ell_n}\times \{1,2,...,K\}\times \mathcal{F}\rightarrow \mathcal{X}^{L},
\end{align} 
i.e., $\hat{W}_k=\psi(A^{[k]}_{1:N},k,\mathsf{F})$ is the retrieved message. 
\end{enumerate}
These functions should satisfy the following two requirements:
\begin{enumerate}

\item \textbf{Correctness: } For any $k\in \{1,2,...,K\}$, $\hat{W}_k=W_k$. 

\item \textbf{Privacy: } For every $k,k'\in\{1,2,...,K\}$, $n\in \{1,2,...,N\}$, and $q\in \mathcal{Q}_n$, 
\begin{align}
\mathbf{Pr}(Q^{[k]}_n=q)=\mathbf{Pr}(Q^{[k']}_n=q).
\end{align}
\end{enumerate}
\end{definition}

It should be noted that $A_n^{[k]}$ is a function of both the messages and the query $Q_n^{[k]}$. Sometimes we need to refer to the answer for a fixed query $Q_n^{[k]}=q$, and this shall be written as $A_n^{(q)}$; the effectiveness of this notation can be seen immediately next. Because of the coding protocol requirement, the overall probability distribution factorizes as follows
\begin{align}
&P\left((Q_n^{[k]},W_{1;K},S_{1:N},A_n^{[k]})=(q,w_{1;K},s_{1:N},a)\right)\notag\\
&=P(Q_n^{[k]}=q)P\left((W_{1;K},S_{1:N})=(w_{1;K},s_{1:N})\right)P\left(A_n^{[k]}=a\Big{|} (Q_n^{[k]},W_{1;K},S_{1:N})=(q, w_{1;K},s_{1:N})\right),
\end{align}
where the conditional distribution is a deterministic one induced by the encoding function $\varphi_{n}$, which can further be simplified to
\begin{align}
P\left(A_n^{[k]}=a\Big{|} (Q_n^{[k]},W_{1;K},S_{1:N})=(q, w_{1;K},s_{1:N})\right)&=P\left(A_n^{[k])}=a\Big{|}(Q_n^{[k]},S_{n})=(q,s_{n})\right)\notag\\
&=P\left(A_n^{(q)}=a\Big{|}(Q_n^{[k]},S_{n})=(q,s_{n})\right).
\end{align}
As a consequence of the privacy requirement and the factorization above, the following joint distributions must also be identical for any $n=1,2,\ldots,N$,
\begin{align}
(A_n^{[k]},Q_n^{[k]},W_{1;K},S_{1:N})\sim (A_n^{[k']},Q_n^{[k']},W_{1;K},S_{1:N}),\quad k,k'\in \{1,2,\ldots,K\},\label{eqn:identical}
\end{align}
which implies that their marginal distributions will also be identical. 
\subsection{Operational and Informational Costs}
The \textit{operational} normalized storage cost for database-$n$ is defined as
\begin{align}
\alpha_n \triangleq \frac{\log_2|\mathcal{S}_n|}{L\log_2 |\mathcal{X}|},\quad n=1,2,\ldots,N, \label{alpha_def}
\end{align}
which is the amount of stored data per bit of individual message; the average storage (per node) cost is then defined as
\begin{align}
\alpha = \frac{1}{N}\sum_{n=1}^N \alpha_n. 
\end{align}
The \textit{operational} normalized download cost for database-$n$ is defined as
\begin{align}
\beta_n \triangleq \frac{\log_2|\mathcal{Y}|\Expt(\ell_n)}{L\log_2 |\mathcal{X}|}, \quad n=1,2,\ldots,N,\label{beta_def}
\end{align}
which is the expected amount of downloaded data per bit desired message at database-$n$, and the average per node download cost is 
\begin{align}
\beta = \frac{1}{N}\sum_{n=1}^N \beta_n. 
\end{align}
Note that $\beta_n$ does not depend on $k$, since the privacy requirement stipulates that the random variable $\ell_n$ has an identical distribution for all $k=1,2,\ldots,K$.

In the literature (e.g., \cite{Sun_Jafar_MPIR}), the \textit{informational} storage cost is sometimes used directly in place of the operational storage cost, i.e., 
\begin{align}
\alpha'_n\triangleq \frac{H(S_n)}{L\log_2 |\mathcal{X}|}, \quad n=1,2,\ldots,N,\label{alphaprime_def}
\end{align}
and correspondingly the average informational storage cost can be defined as
\begin{align}
\alpha' = \frac{1}{N}\sum_{n=1}^N \alpha'_n. 
\end{align}
Similarly the \textit{informational} download cost can be given as
\begin{align}
\beta'_n\triangleq  \frac{ H(A_n^{[k]}|\mathsf{F})}{L\log_2 |\mathcal{X}|}, \quad n=1,2,\ldots,N,\label{betaprime_def}
\end{align}
and the average informational download cost 
\begin{align}
\beta' = \frac{1}{N}\sum_{n=1}^N \beta'_n. 
\end{align}
Once again $\beta'_n$ does not depend on $k$, which can be justified as 
\begin{align}
\frac{ H(A_n^{[k]}|\mathsf{F})}{L\log_2 |\mathcal{X}|}=\frac{ H(A_n^{[k]}|Q_n^{[k]})}{L\log_2 |\mathcal{X}|}=\frac{ H(A_n^{[k']}|Q_n^{[k']})}{L\log_2 |\mathcal{X}|}=\frac{ H(A_n^{[k']}|\mathsf{F})}{L\log_2 |\mathcal{X}|},
\end{align}
where the first and last equality are due the Markov string $\mathsf{F}\leftrightarrow Q_n^{[k]}\leftrightarrow A_n^{[k]}$, and the equality in the middle is due to the privacy condition (\ref{eqn:identical}).

It is clear that 
\begin{align}
\alpha_n \geq \alpha_n\rq{},\quad n=1,2,\ldots,N.
\end{align}
The relation between $\beta_n$ and $\beta_n\rq{}$ is slightly more subtle. It can be seen that for any $n=1,2,\ldots,N$,
\begin{align}
\beta_n = \frac{\log_2|\mathcal{Y}|\Expt(\ell_n)}{L\log_2 |\mathcal{X}|}=\frac{\Expt\left[\Expt(\ell_n|\mathsf{F})\right]\log_2|\mathcal{Y}|}{L\log_2 |\mathcal{X}|}\geq \frac{\Expt\left[H(A^{[k]}_n|\mathsf{F}=f)\right]}{L\log_2 |\mathcal{X}|}=\frac{H(A^{[k]}_n|\mathsf{F})}{L\log_2 |\mathcal{X}|}.
\end{align}
As a consequence, we have 
\begin{align}
\alpha\geq \alpha', \qquad \beta \geq \beta', 
\end{align}
but they may not be equal. In this work, we adopt from a first principle the operational definitions, from which the informational definitions will emerge as the substitutes naturally. 

It was shown \cite{sun2017PIRcapacity} that the minimum download cost is
\begin{align}
\min \beta= \frac{1}{N}+\ldots+\frac{1}{N^{K}}=\frac{N^K-1}{N^K(N-1)},
\end{align}
and codes that can achieve this minimal value are often referred to as optimal private information retrieval codes, or capacity-achieving private information retrieval codes. The codes are optimal in the sense that the download cost is minimal. 
Note that the capacity achieving code given in \cite{sun2017PIRcapacity} assumed fully replicated messages at all databases, however in our setting the databases are not necessarily replicated. 

\section{Extreme Point Characterizations}

We consider the two extreme points in question in the following two subsections, respectively.

\subsection{Minimizing PIR Download Cost without Storage Redundancy}

The first main result we present is for the extreme case when the messages are stored across the databases without any redundancy, i.e., $N\alpha=K$. With such compressed storage, we shall provide a converse to confirm the folklore that the best PIR download strategy is to download everything. This is essentially established through a cut-set-like bound, however, we must apply the privacy condition in the bounding steps, instead of using the cut-set argument directy, to obtain the needed result. 

\begin{theorem}
\label{thm:cutset}
The per-node retrieval cost $\beta$ and per-node storage cost $\alpha$ must satisfy
\begin{align}
(N-1)\alpha +\beta\geq K.
\end{align}
\end{theorem}
\begin{proof}
We start by writing the following inequalities, 
\begin{align}
&\sum_{n'\neq n}\alpha_{n'} L\log_2|\mathcal{X}|+\beta_nL\log_2|\mathcal{X}|\notag\\
&\geq H(S_{1:n-1,n+1:N},A_n^{[1]}|\mathsf{F})\notag\\
&= H(S_{1:n-1,n+1:N},A_n^{[1]},W_1|\mathsf{F})\notag\\
&= L\log_2|\mathcal{X}| + H(S_{1:n-1,n+1:N},A_n^{[1]}|W_1,\mathsf{F})\notag\\
&\stackrel{(*)}{=} L\log_2|\mathcal{X}| + H(S_{1:n-1,n+1:N},A_n^{[2]}|W_1,\mathsf{F}),
\end{align}
where the inequality $(*)$ can be justified as follows
\begin{align}
&H(S_{1:n-1,n+1:N},A_{n}^{[k]}|W_{1:k},\mathsf{F})\notag\\
&=H(S_{1:n-1,n+1:N},A_{n}^{[k]}|W_{1:k},Q_n^{[k]})\notag\\
&=H(S_{1:n-1,n+1:N},A_{n}^{[k+1]}|W_{1:k},Q_n^{[k+1]})\notag\\
&=H(S_{1:n-1,n+1:N},A_{n}^{[k+1]}|W_{1:k},\mathsf{F}),
\end{align}
where the first equality is because of the Markov string $\mathsf{F}\leftrightarrow Q_{n}^{[k]}\leftrightarrow (W_{1:K},S_{1:N},A_n^{[k]})$ and the last equality because of $\mathsf{F}\leftrightarrow Q_{n}^{[k+1]}\leftrightarrow (W_{1:K},S_{1:N},A_n^{[k]+1})$, and the equality in the middle is due to the privacy relation, or more precisely here the identical distribution between
\begin{align}
(S_{1:n-1,n+1:N},W_{1:k},A_{n}^{[k]},Q_n^{[k]})\sim (S_{1:n-1,n+1:N},W_{1:k},A_{n}^{[k+1]},Q_n^{[k+1]}),
\end{align}
due to (\ref{eqn:identical}). 
Continuing to apply the bounding approach on the second term in a similar manner, we eventually arrive at
\begin{align}
\sum_{n'\neq n}\alpha_{n'} L\log_2|\mathcal{X}|+\beta_nL\log_2|\mathcal{X}|\geq K L\log_2|\mathcal{X}|.
\end{align}
Dividing both sides by $L\log_2|\mathcal{X}|$ gives 
\begin{align}
\sum_{n'\neq n}\alpha_{n'}+\beta_n\geq K.\label{eqn:cutsetPIR}
\end{align}
Summing (\ref{eqn:cutsetPIR}) over $n=1,2,\ldots,N$ then normalizing give the desired result.
\end{proof}

\begin{figure}[t]
\centering
\includegraphics[width=0.65\textwidth]{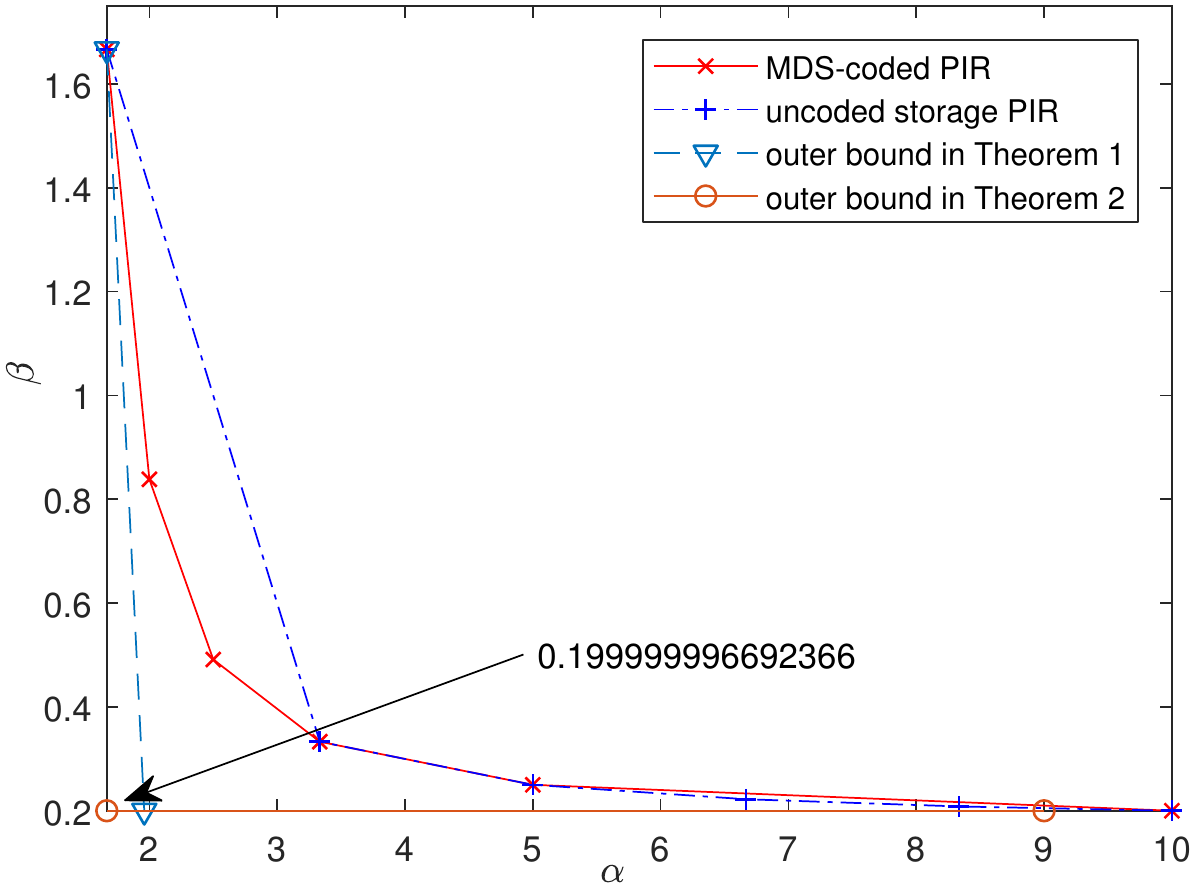}
\caption{Bounds in Theorem \ref{thm:cutset} and Theorem \ref{thm:extreme} when $N=6$ and $K=10$, where the arrow indicates the intercept of the bound in Theorem \ref{thm:extreme} and $\beta=\frac{1}{5}-\frac{1}{5*6^{10}}$ is the minimum download cost.\label{fig:bounds}}
\end{figure}

This bound is illustrated in Fig. \ref{fig:bounds}, together with the achievable tradeoffs with MDS-coded storage \cite{Banawan_Ulukus} and uncoded storage \cite{attia2018capacity}. It is seen that this bound is almost vertical, but it is sufficient to characterize one of two extreme points. The next corollary now follows directly from the theorem, which states that when the storage code has no redundancy, the optimal strategy is to download every message. 
\begin{corollary}
At the minimum storage point $\alpha=\frac{K}{N}$, we must have
\begin{align}
\beta\geq \frac{K}{N}. 
\end{align}
Moreover, the equality can be achieved by downloading all the information from the databases. 
\end{corollary}

\subsection{Minimum Storage Overhead for PIR Capacity-Achieving Codes}

The next main result is a novel lower bound on the storage cost and download cost tradeoff, which leads to an approximate characterization of the extreme point for capacity-achieving codes.
 
\begin{theorem}
\label{thm:extreme}
The per-node retrieval cost $\beta$ and per-node storage cost $\alpha$ when $N\geq 3$ must satisfy
\begin{align}
\frac{\alpha+(N-1) \beta}{N-2}+ N^{K-1}\beta\geq \frac{K}{N-2} +\frac{N^{K}-1}{N(N-1)}.
\end{align}
\end{theorem}

The proof of this theorem can be found in Appendix \ref{appendix:B}. The most important implication of this theorem is the following corollary, which states that for capacity-achieving PIR codes, each database must store at least $K-1$ messages on average. 
\begin{corollary}
At the PIR capacity point $\beta=\frac{N^K-1}{N^K(N-1)}$, we have
\begin{align}
\alpha\geq \left(K-\frac{N^K-1}{N^K}\right)>(K-1). 
\end{align}
\end{corollary}

An illustration of this bound can also be found in Fig. \ref{fig:bounds}. It can be seen that the proposed bound in Theorem \ref{thm:extreme} is almost horizontal, and its intersection with the horizontal axis gives a lower bound on the minimum storage cost when the code is capacity-achieving. Since a trivial storage solution is to replicate all the messages at all the databases, this corollary in fact provides a characterization of the minimum storage cost for capacity-achieving codes within an additive gap of one message. 

In the proof of Theorem \ref{thm:extreme}, the following auxiliary quantities and their relation are important :
\begin{align}
T^{k}&\triangleq H(A_{1:N}^{[k]}|W_{1:k},\mathsf{F}),&\quad k=1,2,\ldots,K,\\
 V_n^{k}&\triangleq H(A_{1:n-1,n+1:N}^{[k]},S_n|W_{1:k},\mathsf{F}),&\quad n=1,2,\ldots,N,\quad k=1,2,\ldots,K,\\
 V^k&=\frac{\sum_{n=1}^NV_n^k}{N},&\quad k=1,2,\ldots,K. 
\end{align}
Due to the definitions above, it is clear that
\begin{align}
T^{K}=V^K=0.
\end{align}

The following two auxiliary lemmas are instrumental to the proof of the theorem, whose proofs can be found in Appendix \ref{appendix:A}.  The first lemma is a recursive relation on $T^k$.
\begin{lemma}
\label{lemma:T}
For any $k=1,2\ldots,K-1$
\begin{align}
T^{k}\geq \frac{L\log_2|\mathcal{X}|}{N}+\frac{T^{k+1}}{N}.\label{eqn:T}
\end{align}
\end{lemma}

The second lemma is a refined recursive relation involving both $V^k$ and $T^k$.
\begin{lemma}
\label{lemma:TR}
For any $k=1,2\ldots,K-1$ and $n=1,2,\ldots,N$,
\begin{align}
\frac{V^{k}}{N-2}+T^{k}\geq \left(\frac{1}{N-2}+\frac{1}{N}\right)L\log_2|\mathcal{X}|+\frac{V^{k+1}}{N-2}+\frac{T^{k+1}}{N}.\label{eqn:TR}
\end{align}
\end{lemma}

Using the recursive relations among $T^k$'s and $V^k$'s in these two lemmas, an induction can be used to prove Theorem \ref{thm:extreme}, which can be found in Appendix \ref{appendix:B}.

\section{An Improved Outer Bound via Pseudo Messages}

In this section, we shall take a closer look at the special case $N=K=2$. Theorem \ref{thm:cutset} in this case specializes to the bound
\begin{align}
\alpha+\beta\geq 2.
\end{align}
As a consequence, the minimum download cost $\beta$ when the storage cost is minimal is clearly $\beta=1$, and this settles one of the two extreme points. Since the messages must be held in the databases as a whole, it is clear that $\alpha\geq 1$ regardless $\beta$, however, the more sophisticated bound in Theorem \ref{thm:extreme} does not apply to $N=2$. The question we wish to address in this section is for this special case $N=K=2$,  whether the other extreme point, where the download cost is minimal, can be more accurately approximated.  For this purpose, we present a novel outer bound based on a pseudo-message technique. 

The main result of this section is the following theorem, which provides an improved outer bound for the storage cost and download cost tradeoff. 

\begin{theorem}
For $N=K=2$, we must have $3\alpha+8\beta \geq 10$.
\label{theorem:NK22}
\end{theorem}

\begin{corollary}
For $N=K=2$, any capacity achieving codes, i.e., codes with $\beta=0.75$, must satisfy $\alpha\geq 4/3$. 
\end{corollary}

This bound is illustrated in Fig. \ref{fig:NK22}, together with the best known achievable tradeoff discovered in \cite{Tian_Sun_Chen_Storage}. It can be seen that this new bound can indeed improve the accuracy of the approximation for the extreme point on the horizontal axis. The proof of this bound is rather technical, the details of which are relegated to Appendix \ref{appendix:C}, but the main proof idea is given and discussed in the remainder of this section. The new bound in Theorem \ref{theorem:NK22} appears difficult to generalize to larger values of $N$ and $K$. Nevertheless, it can be viewed as a strong piece of evidence that the outer bounds we have at this point are likely not tight, and more sophisticated techniques involving non-Shannon type inequalities may provide additional improvement. In fact, without using the pseudo-message technique, we have indeed applied the computational approach \cite{tian2014characterizing,tian2019open} on this two-message two-database problem, i.e., invoking all Shannon type inequalities, which did not produce any bound stronger than that in Theorem \ref{thm:cutset}.

\begin{figure}[tb]
\centering
\includegraphics[width=0.65\textwidth]{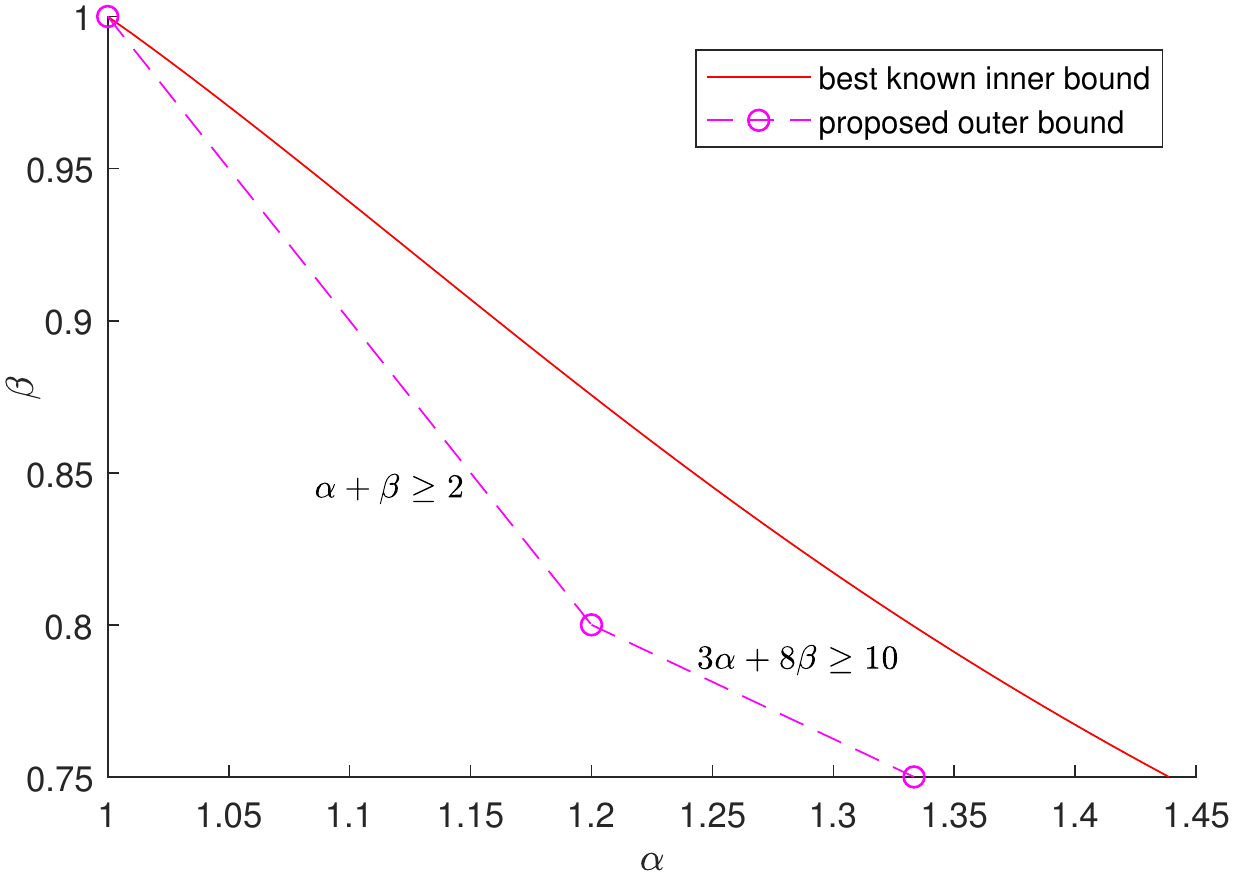}
\caption{The proposed outer bound and the best known inner bound when $N=K=2$.\label{fig:NK22}}
\end{figure}

In order to prove this bound, we need to first introduce some  necessary simplifications on the code structure used in the derivation. With such simplification, we shall discuss the main pseudo-message technique in some more details.

\subsection{Symmetry Assumptions}

It was shown in \cite{tian2018capacity} that there are three types of symmetry relations in this private information retrieval problem. By applying the three types of symmetry relations, it can be shown that it is without loss of optimality to consider only codes such that 
\begin{align}
&H(S_n)=H(S_{n'}),&\quad n,n\rq{}\in\{1,2,\ldots,N\}\notag\\
&H(A_n^{(q)})=H(A_{n\rq{}}^{(q\rq{})}),&\quad n,n\rq{}\in\{1,2,\ldots,N\},\quad q\in \mathcal{Q}_n,\quad q\rq{}\in\mathcal{Q}_{n\rq{}}\notag\\
&H(A_n^{(q)},W_k)=H(A_{n'}^{(q')},W_{k\rq{}}),&\quad n,n\rq{}\in\{1,2,\ldots,N\},\quad q\in \mathcal{Q}_n,\quad q\rq{}\in\mathcal{Q}_{n\rq{}}, \quad k,k\rq{}\in\{1,2,\ldots,K\}.\label{eqn:sym}
\end{align}
In other words, the databases use the same amount of storage, the answers all have the same entropy, and the combinations of any single answer and any single message all have the same joint entropy. Moreover, for such symmetrized codes, we also have 
\begin{align}
\beta=\beta_n, \quad n=1,2,\ldots,N.
\end{align}
As a consequence, we have 
\begin{align}
\beta\geq H(A_n^{(q)}), \quad q\in\mathcal{Q}_n,\quad n=1,2,\ldots,N.
\end{align}
The symmetry in terms of the joint entropy can be extended beyond that in (\ref{eqn:sym}), however in this work there is no need for such generality. The discussion in the following in effect utilizes the concept of answer variety introduced in \cite{tian2018capacity}, though we will not explicitly invoke the concept, and the discussion will be self-contained.

\subsection{A Subtle Dependence Structure}

Our first step is to consider the relation among different answers:
\begin{enumerate}
\item First consider an answer at database-$1$ for an arbitrary but fixed $q_1\in\mathcal{Q}_1$ which can be used to retrieve $W_1$, denoted as $X_1=A^{(q_1)}_1$. Clearly there exists a query $q_2\in \mathcal{Q}_2$ such that together with the answer $Y_1=A_2^{(q_2)}$, the message $W_1$ can recovered, i.e., $H(W_1|X_1,Y_1)=0$.
\item Because of the privacy constraint, the answer $X_1=A_1^{(q_1)}$ from database-1 can be used, together with an query $q_2\rq{}\in\mathcal{Q}_2$, to retrieve $W_2$, i.e., with the answer $Y_2=A^{(q_2\rq{})}_2$ we have $H(W_2|X_1,Y_2)=0$. Note that the answer $Y_1$ and $Y_2$ are not necessarily distinct.
\item Continuing the same argument, an answer $X_2=A_1^{(q_1\rq{})}$ must exist such that $H(W_2|X_2,Y_1)=0$.
\item Finally, an answer $X_3=A_1^{(q_1'')}$ must also exist such that $H(W_1|X_3,Y_2)=0$.
\end{enumerate}
Clearly, $(X_1,X_2,X_3,Y_1,Y_2)$ are functions of $W_1,W_2$, i.e., $H(X_1,X_2,X_3,Y_1,Y_2|W_1,W_2)=0$. 

Since the answers must come from the stored content, $H(X_1,X_2,X_3|S_1)=0$, we must have
\begin{align}
\alpha\geq H(S_1)\geq H(X_1,X_2,X_3),
\end{align}
and similarly
\begin{align}
\alpha\geq H(S_2)\geq H(Y_1,Y_2).
\end{align}
This is the dependence structure and constraints in the original problem setting that we shall avail. Note that in this line of proof, the random key $\mathsf{F}$ is not playing any significant role, unlike in the proofs for Theorem \ref{thm:cutset} and \ref{thm:extreme}. 
  
\subsection{A Pseudo Message Technique}

Our next step is to extend the random variable space, by introducing some random variables not in the original problem. 
The random variables $V_1,V_2$, referred to as the pseudo messages, are introduced into the setting using the so-called Markov coupling 
\begin{align}
(V_1,V_2)\leftrightarrow (Y_1,Y_2)\leftrightarrow (W_1,W_2,X_1,X_2,X_3). \label{eqn:markov1}
\end{align}
Moreover, the two sets of random variables have the identical marginal distribution
\begin{align}
(Y_1,Y_2,V_1,V_2)\sim (Y_1,Y_2,W_1,W_2).\label{eqn:mirror1}
\end{align}
Similarly, the second set of pseudo message random variables $(U_1,U_2)$ are also introduced such that 
\begin{align}
(U_1,U_2)\leftrightarrow (X_1,X_2,X_3)\leftrightarrow (W_1,W_2,Y_1,Y_2,V_1,V_2),\label{eqn:markov2}
\end{align}
and the two sets of random variables have the identical marginal distribution
\begin{align}
(X_1,X_2,X_3,U_1,U_2)\sim (X_1,X_2,X_3,W_1,W_2). \label{eqn:mirror2}
\end{align}
As a consequence, the extended set of random variables can be factorized as 
\begin{align}
P(W_1,W_2,X_1,X_2,X_3,Y_1,Y_2)P(U_1,U_2|X_1,X_2,X_3)P(V_1,V_2|Y_1,Y_2).
\end{align}

The proof of Theorem \ref{theorem:NK22} utilizes the symmetry, the Markov condition in the extended random variable space, as well as the encoding and decoding constraints. The basic idea is to bound or substitute the conditional entropy involving $(W_1,W_2,X_1,X_2,X_3,Y_1,Y_2)$ using conditional entropy involving subsets of $(U_1,U_2,V_1,V_2,W_1,W_2,X_1,X_2,X_3,Y_1,Y_2)$ in matching forms, and then to cancel terms using the identical distribution relations. 

The proof technique of introducing pseudo messages can be viewed as being closely related to non-Shannon type inequalities, since all known non-Shannon type inequalites are essentially produced by introducing certain mirrored copies. This proof was obtained with the assistance of the computational tool that the author developed previously \cite{tian2019open}, which was found valuable in the investigation of the regenerating code problem \cite{tian2014characterizing} and the coded caching problem \cite{tian2018symmetry}.  The main difference between the technique we use in this work and those seen in generating non-Shannon type inequalities is that instead of introducing a single-sided mirrored set, we introduce mirrors on two sides--one side being $(V_1,V_2)$, and the other being $(U_1,U_2)$--both of which through the Markov coupling.

\section{Conclusion}

We initiated the investigation of the fundamental tradeoff between the storage cost and the download cost in general private information retrieval systems. Several novel outer bounds are provided. On the one hand, we were able to confirm the folklore that when the messages are stored without any redundancy, the retrieval must download all the messages. On the other hand, when the code is PIR capacity-achieving, we establish the somewhat surprising result that the storage cost cannot be too much lower than storing all the messages at each database. Moreover, we show that for the two-message two-database case, a more elaborate pseudo-message technique can be used to derive a stronger outer bound. As an ongoing work, we are investigating more general outer bounds of the fundamental tradeoff between the storage cost and the download cost.

\section*{Acknowledgment}

The author wishes to thank Dr. Hua Sun for several discussions and for providing comments on an early draft of the paper. 

\appendix
\section{Proof of Lemmas  \ref{lemma:T} and  \ref{lemma:TR}}
\label{appendix:A}

\begin{proof}[Proof of Lemma \ref{lemma:T}]
We write the following chain of inequalities
\begin{align}
T^{k}&=H(A_{1:N}^{[k]}|W_{1:k},\mathsf{F})\notag\\
&\geq \frac{1}{N}\sum_{n=1}^N H(A_{n}^{[k]}|W_{1:k},\mathsf{F})\notag\\
&\stackrel{(*)}{=} \frac{1}{N}\sum_{n=1}^N H(A_{n}^{[k+1]}|W_{1:k},\mathsf{F})\notag\\
&\geq \frac{1}{N}H(A_{1:N}^{[k+1]}|W_{1:k},\mathsf{F})\notag\\
&= \frac{1}{N}H(A_{1:N}^{[k+1]},W_{k+1}|W_{1:k},\mathsf{F})\notag\\
&\geq \frac{L\log_2|\mathcal{X}|}{N}+\frac{1}{N}H(A_{1:N}^{[k+1]}|W_{1:k+1},\mathsf{F}). \label{eqn:capacity}
\end{align}
The inequality $(*)$ can be justified as follows
\begin{align}
H(A_{n}^{[k]}|W_{1:k},\mathsf{F})=H(A_{n}^{[k]}|W_{1:k},Q_n^{[k]})=H(A_{n}^{[k+1]}|W_{1:k},Q_n^{[k+1]})=H(A_{n}^{[k+1]}|W_{1:k},\mathsf{F}),
\end{align}
where the first equality is because of the Markov string $\mathsf{F}\leftrightarrow Q_{n}^{[k]}\leftrightarrow (W_{1:K},S_{1:N},A_n^{[k]})$ and the last equality because of $\mathsf{F}\leftrightarrow Q_{n}^{[k+1]}\leftrightarrow (W_{1:K},S_{1:N},A_n^{[k]+1})$, and the equality in the middle is due to the privacy relation, i.e.,
\begin{align}
(W_{1:k},A_{n}^{[k]},Q_n^{[k]})\sim (W_{1:k},A_{n}^{[k+1]},Q_n^{[k+1]}),
\end{align}
because of (\ref{eqn:identical}). 
This $(*)$ notation will also be used in the rest of the paper. The proof is now complete. 
\end{proof}

\begin{proof}[Proof of Lemma \ref{lemma:TR}]
We first notice that 
\begin{align}
V_n^{k}&=H(A_{1:n-1,n+1:N}^{[k]},S_{n}|W_{1:k},\mathsf{F})\notag\\
&=H(S_{n}|\mathsf{F},W_{1:k})+H(A_{1:n-1,n+1:N}^{[k]}|S_{n},W_{1:k},\mathsf{F})\notag\\
&\geq H(S_{n}|\mathsf{F},W_{1:k})+\frac{1}{N-1}\sum_{n'\neq n} H(A_{n'}^{[k]}|S_{n},W_{1:k},\mathsf{F})\\
&\stackrel{(*)}{=} H(S_{n}|\mathsf{F},W_{1:k})+\frac{1}{N-1}\sum_{n'\neq n} H(A_{n'}^{[k+1]}|S_{n},W_{1:k},\mathsf{F})\\
&\geq H(S_{n}|\mathsf{F},W_{1:k})+\frac{1}{N-1}  H(A_{1:n-1,n+1:N}^{[k+1]}|S_{n},W_{1:k},\mathsf{F})\\
&= H(S_{n}|F,W_{1:k})+\frac{1}{N-1}  H(A_{1:n-1,n+1:N}^{[k+1]},W_{k+1}|S_{n},W_{1:k},\mathsf{F})\\
&= \frac{N-2}{N-1}H(S_{n}|\mathsf{F},W_{1:k})+\frac{1}{N-1} H(A_{1:n-1,n+1:N}^{[k+1]},W_{k+1},S_{n}|W_{1:k+1},\mathsf{F})\notag\\
&= \frac{N-2}{N-1}H(S_{n}|\mathsf{F},W_{1:k})+\frac{L\log_2|\mathcal{X}|}{N-1}+\frac{1}{N-1} H(A_{1:n-1,n+1:N}^{[k+1]},S_{n}|W_{1:k+1},\mathsf{F})\notag\\
&=\frac{N-2}{N-1}H(S_{n}|\mathsf{F},W_{1:k})+\frac{L\log_2|\mathcal{X}|}{N-1}+\frac{1}{N-1} V_n^{k+1}.
\end{align}

With the inequality above, we can then write
\begin{align}
\frac{V_n^{k}}{N-2}+T^{k}&\geq H(A_{1:N}^{[k]}|W_{1:k},\mathsf{F})+\frac{1}{N-1}H(S_{n}|W_{1:k},\mathsf{F})+\frac{L\log_2|\mathcal{X}|+V_n^{k+1}}{(N-1)(N-2)}\notag\\
&=H(A_{1:N}^{[k]}|W_{1:k},\mathsf{F})+\frac{1}{N-1}H(S_{n},A_{n}^{[k]}|W_{1:k},\mathsf{F})+\frac{L\log_2|\mathcal{X}|+V_n^{k+1}}{(N-1)(N-2)}\notag\\
&=\frac{N-2}{N-1}H(A_{1:N}^{[k]}|W_{1:k},\mathsf{F})+\frac{2}{N-1}H(A_{n}^{[k]}|W_{1:k},\mathsf{F})\notag\\
&\qquad+\frac{1}{N-1}\bigg{(}H(A_{1:n-1,n+1:N}^{[k]}|A_{n}^{[k]},W_{1:k},\mathsf{F})+ H(S_{n}|A_{n}^{[k]},W_{1:k},\mathsf{F})\bigg{)}+\frac{L\log_2|\mathcal{X}|+V_n^{k+1}}{(N-1)(N-2)}\notag\\
&\geq \frac{N-2}{N-1}H(A_{1:N}^{[k]}|W_{1:k},\mathsf{F})+\frac{2}{N-1}H(A_{n}^{[k]}|W_{1:k},\mathsf{F})\notag\\
&\qquad+\frac{1}{N-1}H(A_{1:n-1,n+1:N}^{[k]},S_{n}|A_{n}^{[k]},W_{1:k},\mathsf{F})+\frac{L\log_2|\mathcal{X}|+V_n^{k+1}}{(N-1)(N-2)}\notag\\
&=\frac{N-2}{N-1}H(A_{1:N}^{[k]}|W_{1:k},\mathsf{F})+\frac{1}{N-1}H(A_{n}^{[k]}|W_{1:k},\mathsf{F})\notag\\
&\qquad+\frac{1}{N-1}H(S_{n}|W_{1:k},\mathsf{F})+\frac{1}{N-1}H(A_{1:n-1,n+1:N}^{[k]}|S_{n},W_{1:k},\mathsf{F})+\frac{L\log_2|\mathcal{X}|+V_n^{k+1}}{(N-1)(N-2)}, \label{eqn:recur1}
\end{align}
where the first term and third term have a similar form as the first and the second term at the beginning of the chain, only with slightly different coefficients. Continuing the same manipulation as in (\ref{eqn:recur1}), we arrive at
\begin{align}
&\frac{V_n^{k}}{N-2}+T^{k}\notag\\
&\geq H(A_{n}^{[k]}|W_{1:k},\mathsf{F})+H(A_{1:n-1,n+1:N}^{[k]}|S_{n},W_{1:k},\mathsf{F})+\frac{1}{N-1}H(S_{n}|W_{1:k},\mathsf{F})+\frac{L\log_2|\mathcal{X}|+V_n^{k+1}}{(N-1)(N-2)}\notag\\
&\geq H(A_{n}^{[k]}|W_{1:k},\mathsf{F})+\frac{1}{N-1}\sum_{n'\neq n}^{N-1} H(A_{n'}^{[k]}|S_{n},W_{1:k},\mathsf{F})+\frac{1}{N-1}H(S_{n}|W_{1:k},\mathsf{F})+\frac{L\log_2|\mathcal{X}|+V_n^{k+1}}{(N-1)(N-2)}\notag\\
&\stackrel{(*)}{=}H(A_{n}^{[k+1]}|W_{1:k},\mathsf{F})+\frac{1}{N-1}\sum_{n'\neq n}^{N-1} H(A_{n'}^{[k+1]}|S_{n},W_{1:k},\mathsf{F})+\frac{1}{N-1}H(S_{n}|W_{1:k},\mathsf{F})+\frac{L\log_2|\mathcal{X}|+V_n^{k+1}}{(N-1)(N-2)}\notag\\
&\geq H(A_{n}^{[k+1]}|W_{1:k},\mathsf{F})+\frac{1}{N-1} H(A_{1:n-1,n+1:N}^{[k+1]},S_{n}|W_{1:k},\mathsf{F})+\frac{L\log_2|\mathcal{X}|+V_n^{k+1}}{(N-1)(N-2)}\notag\\
&= H(A_{n}^{[k+1]}|W_{1:k},\mathsf{F})+\frac{1}{N-1} H(A_{1:n-1,n+1:N}^{[k+1]},S_{n},W_{k+1}|W_{1:k},\mathsf{F})+\frac{L\log_2|\mathcal{X}|+V_n^{k+1}}{(N-1)(N-2)}\notag\\
&\geq H(A_{n}^{[k+1]}|W_{1:k},\mathsf{F})+\frac{L\log_2|\mathcal{X}|}{N-2}+\frac{V_n^{k+1}}{N-2}.\label{eqn:beforesumN}
\end{align}

Now summing  (\ref{eqn:beforesumN}) over $n=1,2,\ldots,N$ and then taking the average, we arrive at
\begin{align}
\frac{V^{k}}{N-2}+T^{k} &\geq \frac{1}{N}\sum_{n=1}^{N}H(A_{n}^{[k+1]}|W_{1:k},\mathsf{F})+\frac{L\log_2|\mathcal{X}|}{N-2}+\frac{V^{k+1}}{N-2}\notag\\
&\geq \frac{1}{N}H(A_{1:N}^{[k+1]}|W_{1:k},\mathsf{F})+\frac{L\log_2|\mathcal{X}|}{N-2}+\frac{V^{k+1}}{N-2}\notag\\
&=\frac{L\log_2|\mathcal{X}|}{N}+\frac{1}{N}T^{k+1}+\frac{L\log_2|\mathcal{X}|}{N-2}+\frac{V^{k+1}}{N-2}.
\end{align}
The proof is now complete. 
\end{proof}

\section{Proof of Theorem \ref{thm:extreme}}
\label{appendix:B}
\begin{proof}[Proof of Theorem \ref{thm:extreme}]
We first prove the following bound by induction
\begin{align}
&\frac{V^{k}}{N-2}+N^{K-k-1}T^{k} \geq \frac{N^{K-k}-1}{N(N-1)}L\log_2|\mathcal{X}|+\frac{(K-k)L\log_2|\mathcal{X}|}{N-2},\qquad k=1,2,\ldots,K-1. \label{eqn:ind2}
\end{align}
For this purpose, first consider $k=K-1$, which is simply Lemma \ref{lemma:TR} since $T^{K}=V^K=0$. Now assume the claim is true for $k=k^*$, and we next show that it is true for $k=k^*-1$. For this purpose we write
\begin{align}
\frac{V^{k^*-1}}{N-2}+N^{K-k^*}T^{k^*-1}&=\frac{V^{k^*-1}}{N-2}+T^{k^*-1}+(N^{K-k^*}-1)T^{k^*-1}\notag\\
&\geq \left(\frac{1}{N-2}+\frac{1}{N}\right)L\log_2|\mathcal{X}|+\frac{V^{k^*}}{N-2}+\frac{T^{k^*}}{N}+(N^{K-k^*}-1)\left(\frac{L\log_2|\mathcal{X}|}{N}+\frac{1}{N}T^{k^*}\right)\notag\\
&= \frac{V^{k^*}}{N-2}+N^{K-k^*-1}T^{k^*}+N^{K-k^*-1}L\log_2|\mathcal{X}|+\frac{L\log_2|\mathcal{X}|}{N-2},
\end{align}
where the inequality is by applying Lemma \ref{lemma:T} and Lemma \ref{lemma:TR}. By the assumption that (\ref{eqn:ind2}) holds for $k=k^*$, we have
\begin{align}
\frac{V^{k^*-1}}{N-2}+N^{K-k^*}T^{k^*-1}&\geq \frac{N^{K-k^*}-1}{N(N-1)}L\log_2|\mathcal{X}|+\frac{(K-k^*)L\log_2|\mathcal{X}|}{N-2}+N^{K-k^*-1}L\log_2|\mathcal{X}|+\frac{L\log_2|\mathcal{X}|}{N-2}\notag\\
&=\frac{N^{K-k^*+1}-1}{N(N-1)}L\log_2|\mathcal{X}|+\frac{(K-k^*+1)L\log_2|\mathcal{X}|}{N-2}, 
\end{align}
which is the desired inequality for $k=k^*-1$. 

The bound stated in the theorem can now be obtained by taking $k=1$ in (\ref{eqn:ind2})
\begin{align}
\frac{V^{1}}{N-2}+N^{K-2}T^{1} \geq \frac{N^{K-1}-1}{N(N-1)}L\log_2|\mathcal{X}|+\frac{(K-1)L\log_2|\mathcal{X}|}{N-2}\label{eqn:TV12TVK}
\end{align}
and noticing that by the definition of $\alpha$ and $\beta$
\begin{align}
&\frac{\alpha+(N-1) \beta}{N-2}L\log_2|\mathcal{X}|+ N^{K-1}\beta L\log_2|\mathcal{X}|\notag\\
&\geq\frac{\sum_{n=1}^NH(A^{[1]}_{1:n-1,n+1:N},S_{n}|\mathsf{F})}{N(N-2)}+N^{K-2}H(A^{[1]}_{1:N}|\mathsf{F})\notag\\
&=\frac{\sum_{n=1}^NH(A^{[1]}_{1:n-1,n+1:N},S_{n},W_1|\mathsf{F})}{N(N-2)}+N^{K-2}H(A^{[1]}_{1:N},W_1|\mathsf{F})\notag\\
&=\frac{L\log_2|\mathcal{X}|}{N-2}+N^{K-2}L\log_2|\mathcal{X}|+\frac{V^{1}}{N-2}+N^{K-2}T^{1}\notag\\
&\geq \frac{L\log_2|\mathcal{X}|}{N-2}+N^{K-2}L\log_2|\mathcal{X}|+\frac{N^{K-1}-1}{N(N-1)}L\log_2|\mathcal{X}|+\frac{(K-1)L\log_2|\mathcal{X}|}{N-2}\notag\\
&=\frac{KL\log_2|\mathcal{X}|}{N-2}+\frac{N^{K}-1}{N(N-1)}L\log_2|\mathcal{X}|,
\end{align}
where the last inequality is due to (\ref{eqn:TV12TVK}). 
Dividing both sides by $L\log_2|\mathcal{X}|$ completes the proof. 
\end{proof}

\section{Proof of Theorem \ref{theorem:NK22}}
\label{appendix:C}
\begin{proof}[Proof of \ref{theorem:NK22}]
We start by
\begin{align}
6\alpha L\log_2|\mathcal{X}|+16\beta L\log_2|\mathcal{X}|&\geq 3H(S_1)+3H(S_2)+8H(X_1)+8H(Y_2)\notag\\
&\geq 3H(X_1,X_2,X_3)+3H(Y_1,Y_2)+8H(X_1,Y_2).
\end{align}
Note that
\begin{align}
H(Y_1,Y_2)&=H(Y_1,Y_2,V_1,V_2)-H(V_1,V_2|Y_1,Y_2)\nonumber\\
&\stackrel{(m)}{=}H(Y_1,Y_2,V_1,V_2)-H(V_1,V_2|W_1,W_2,Y_1,Y_2)\nonumber\\
&\stackrel{(f)}{=}H(Y_1,Y_2,V_1,V_2)-H(V_1,V_2|W_1,W_2)\nonumber\\
&\stackrel{(i)}{=}H(Y_1,Y_2,W_1,W_2)-H(V_1,V_2|W_1,W_2)\nonumber\\
&\stackrel{(f)}{=}H(W_1,W_2)-H(V_1,V_2|W_1,W_2)\nonumber\\
&=2L\log_2|\mathcal{X}|-H(V_1,V_2|W_1,W_2),
\end{align}
where $(m)$ means by the Markov string relation (\ref{eqn:markov1}), $(f)$ means because of the coding function relation $H(Y_1,Y_2|W_1,W_2)=0$, and $(i)$ means the identical (i.e., mirrored) distribution (\ref{eqn:mirror1}). We will also use $(m)$, $(f)$, and $(i)$ in the sequel to indicate the justifications for the same (or similar) reasons. 
Similarly we can write
\begin{align}
H(X_1,X_2,X_3)&=H(X_1,X_2,X_3,U_1,U_2)-H(U_1,U_2|X_1,X_2,X_3)\nonumber\\
&\stackrel{(m)}{=}H(X_1,X_2,X_3,U_1,U_2)-H(U_1,U_2|X_1,X_2,X_3,W_1,W_2,V_1,V_2)\nonumber\\
&\stackrel{(f)}{=}H(X_1,X_2,X_3,U_1,U_2)-H(U_1,U_2|W_1,W_2,V_1,V_2)\nonumber\\
&\stackrel{(i)}{=}H(X_1,X_2,X_3,W_1,W_2)-H(U_1,U_2|W_1,W_2,V_1,V_2)\nonumber\\
&=H(W_1,W_2)-H(U_1,U_2|W_1,W_2,V_1,V_2)\nonumber\\
&=2L\log_2|\mathcal{X}|-H(U_1,U_2|W_1,W_2,V_1,V_2).
\end{align}
It follows that
\begin{align}
&6\alpha L\log_2|\mathcal{X}|+16\beta L\log_2|\mathcal{X}|\notag\\
&\geq 3H(X_1,X_2,X_3)+3H(Y_1,Y_2)+8H(X_1,Y_2)\nonumber\\
&\geq 12L\log_2|\mathcal{X}|-3H(V_1,V_2|W_1,W_2)-3H(U_1,U_2|W_1,W_2,V_1,V_2)+8H(X_1,Y_2)\nonumber\\
&\stackrel{(f)}{=}18L\log_2|\mathcal{X}|-3H(U_1,U_2,W_1,W_2,V_1,V_2)+8H(X_1,Y_2,W_2).
\end{align}

Now we wish to upper bound the second term
\begin{align}
&H(U_1,U_2,W_1,W_2,V_1,V_2)\notag\\
&=H(W_1,W_2,V_1,V_2,U_2)+H(U_1|W_1,W_2,V_1,V_2,U_2)\nonumber\\
&\stackrel{(f)}{=}H(W_1,W_2,V_1,V_2,U_2)+H(U_1,X_2|X_1,X_3, W_1,W_2,V_1,V_2,U_2)\nonumber\\
&\leq H(W_1,W_2,V_1,V_2,U_2)+H(U_1,X_2|X_1,X_3,U_2)\nonumber\\
&= H(W_1,W_2,V_2,U_2)+H(V_1|W_1,W_2,V_2,U_2)+H(U_1,X_2|X_1,X_3,U_2)\nonumber\\
&\leq H(W_1,W_2,V_2,U_2)+H(V_1|Y_1,Y_2,V_2)+H(U_1,X_2|X_1,X_3,U_2)\nonumber\\
&\leq H(W_1,W_2,V_2,U_2)+H(V_1,V_2,Y_1,Y_2)-H(Y_1,Y_2,V_2)+H(U_1,U_2,X_1,X_2,X_3)-H(X_1,X_3,U_2)\nonumber\\
&\stackrel{(i,f)}{=} H(W_1,W_2,V_2,U_2)+H(W_1,W_2)-H(Y_1,Y_2,V_2)+H(W_1,W_2)-H(X_1,X_3,U_2)\nonumber\\
&=H(W_1,W_2,V_2,U_2)+4L\log_2|\mathcal{X}|-H(Y_1,Y_2,V_2)-H(X_1,X_3,U_2). 
\end{align}
Thus we have
\begin{align}
&6\alpha L\log_2|\mathcal{X}|+16\beta L\log_2|\mathcal{X}|\notag\\
&\qquad\geq 6L\log_2|\mathcal{X}|-3H(W_1,W_2,V_2,U_2)+3H(Y_1,Y_2,V_2)+3H(X_1,X_3,U_2)+8H(X_1,Y_2,W_2). \label{eqn:separate}
\end{align}

We bound the last three terms as
\begin{align}
&3H(Y_1,Y_2,V_2)+3H(X_1,X_3,U_2)+8H(X_1,Y_2,W_2)\nonumber\\
&\geq  3[H(Y_1,Y_2,V_2)+H(X_1,Y_2,W_2)]+3[H(X_1,X_3,U_2)+H(X_1,Y_2,W_2)]+2H(X_1,W_2)\nonumber\\
&\stackrel{(i)}{=}  3[H(Y_1,Y_2,W_2)+H(X_1,Y_2,W_2)]+3[H(X_1,X_3,W_2)+H(X_1,Y_2,W_2)]+2H(X_1,W_2)\nonumber\\
&= 3[2H(Y_2,W_2)+H(Y_1|Y_2,W_2)+H(X_1|Y_2,W_2)]\notag\\
&\qquad+3[2H(X_1,W_2)+H(X_3|X_1,W_2)+H(Y_2|X_1,,W_2)]+2H(X_1,W_2)\nonumber\\
&\geq  3[2H(Y_2,W_2)+H(Y_1,X_1|Y_2,W_2)]+3[2H(X_1,W_2)+H(X_3,Y_2|X_1,W_2)]+2H(X_1,W_2)\nonumber\\
&=3[H(Y_2,W_2)+H(Y_1,X_1,Y_2,W_2)]+3[H(X_1,W_2)+H(X_3,Y_2,X_1,W_2)]+2H(X_1,W_2)\nonumber\\
&\stackrel{(f)}{=}3[H(Y_2,W_2)+H(Y_1,X_1,Y_2,W_1,W_2)]+3[H(X_1,W_2)+H(X_3,Y_2,X_1,W_1,W_2)]+2H(X_1,W_2)\nonumber\\
&=  12L\log_2|\mathcal{X}|+3H(Y_2,W_2)+3H(X_1,W_2)+2H(X_1,W_2)\notag\\
&\stackrel{(s)}{=}12L\log_2|\mathcal{X}|+8H(X_1,W_2),
\end{align} 
where $(s)$ is due to the symmetry relation (\ref{eqn:sym}); we will continue to use $(s)$ to indicate the same justification. 
The second term in (\ref{eqn:separate}) needs to be upper-bounded, which is given as
\begin{align}
&3H(W_1,W_2,V_2,U_2)\nonumber\\
&\stackrel{(f)}{=}[H(Y_2,W_1,V_2,U_2)+H(X_1|Y_2,W_1,V_2,U_2)]+[H(Y_1,W_2,V_2,U_2)+H(X_1|Y_1,W_2,V_2,U_2)]\nonumber\\
&\qquad+[H(W_1,V_2,U_2)+H(W_2|W_1,V_2,U_2)]\nonumber\\
&\leq [H(Y_2,W_1,V_2,U_2)+H(X_1|W_1,V_2,U_2)]+[H(Y_1,W_2,V_2,U_2)+H(X_1|W_2,V_2,U_2)]\nonumber\\
&\qquad+[H(W_1,V_2,U_2)+H(W_2|V_2,U_2)]\nonumber\\
&=H(Y_2,W_1,V_2,U_2)+H(X_1,W_1,V_2,U_2)+H(Y_1,W_2,V_2,U_2)+H(X_1,W_2,V_2,U_2)-H(V_2,U_2)\nonumber\\
&\leq H(Y_2,X_3,W_1,V_2,U_2)+H(X_1,Y_1,W_1,V_2,U_2)+H(X_2,Y_1,W_2,V_2,U_2)\nonumber\\
&\qquad+H(X_1,Y_2,W_2,V_2,U_2)-H(V_2,U_2).
\end{align}

Thus we have
\begin{align}
&6\alpha L\log_2|\mathcal{X}|+16\beta L\log_2|\mathcal{X}|\notag\\
&\geq 18L\log_2|\mathcal{X}|+8H(X_1,W_2)+H(V_2,U_2)-[H(Y_2,X_3,W_1,V_2,U_2)+H(X_1,Y_1,W_1,V_2,U_2)\nonumber\\
&\qquad\qquad+H(X_2,Y_1,W_2,V_2,U_2)+H(X_1,Y_2,W_2,V_2,U_2)]\label{eqn:4ways}
\end{align}

Notice 
\begin{align}
&H(Y_2,X_3,W_1,V_2,U_2)-H(X_1,W_2)\notag\\
&\stackrel{(s)}{=}H(Y_2,X_3,W_1,V_2,U_2)-H(X_3,W_2)\notag\\
&\stackrel{(i)}{=}H(Y_2,X_3,W_1,V_2,U_2)-H(X_3,U_2)\notag\\
&=H(Y_2,V_2|X_3,U_2)\leq H(Y_2,V_2|U_2).
\end{align}

We can similarly bound the other terms in (\ref{eqn:4ways}), and arrive at
\begin{align}
&6\alpha L\log_2|\mathcal{X}|+16\beta L\log_2|\mathcal{X}|\nonumber\\
&\geq 18 L\log_2|\mathcal{X}|+4H(X_1,W_2)+H(V_2,U_2)-H(Y_2,V_2|U_2)-H(Y_1,V_2|U_2)-H(Y_1,V_2|U_2)-H(Y_2,V_2|U_2)\nonumber\\
&\geq 22 L\log_2|\mathcal{X}|+4H(X_1,W_2)+H(V_2,U_2)-2H(Y_2,V_2,U_2)-2H(Y_1,V_2,U_2)\nonumber\\
&\stackrel{(s)}{=} 22 L\log_2|\mathcal{X}|+2H(Y_1,V_2)+2H(Y_2,V_2)+H(V_2,U_2)-2H(Y_2,V_2,U_2)-2H(Y_1,V_2,U_2)\nonumber\\
&= 22 L\log_2|\mathcal{X}|+H(V_2,U_2)-2H(U_2|Y_2,V_2)-2H(U_2|Y_1,V_2)\nonumber\\
&\geq 22 L\log_2|\mathcal{X}|+H(V_2,U_2)-4H(U_2|V_2)\nonumber\\
&\geq 22 L\log_2|\mathcal{X}|+H(V_2,U_2)-H(U_2|V_2)-3H(U_2)\nonumber\\
&=22 L\log_2|\mathcal{X}|+H(V_2,U_2)-H(V_2)-H(U_2|V_2)-2H(U_2)\geq 20L\log_2|\mathcal{X}|,
\end{align}
where the second equality and the last equality are due to the identical distribution of $U_2$ and $V_2$, which are both identically distributed to $W_2$. 
Normalizing both sides gives the stated result. 
\end{proof}

\bibliographystyle{IEEEtran}
 \newcommand{\noop}[1]{}

\end{document}